\documentclass[copyright,creativecommons,noderivs,noncommercial]{eptcs}
\providecommand{\event}{PLACES 2010} 
\providecommand{\firstpage}{1} 

\usepackage{amsmath,amsfonts,amssymb,proof,amsthm} %
\usepackage{graphicx}
\usepackage[inference,ligature]{semantic}
\usepackage{macros}
\usepackage{url}
\usepackage{xspace}

\graphicspath{{graphics/}{img/}}

\urldef{\mailsa}\path|{carbonem,|
\urldef{\mailsb}\path|dgro,hilde,lopez}@itu.dk|

\begin{document}

\title{A Logic for Choreographies\thanks{The authors are listed in
    alphabetical order.}}

\def\titlerunning{A logic for Choreographies}
\def\authorrunning{M.~Carbone \& D.~Grohmann \& T.~Hildebrandt \& H.~L\'{o}pez}

\author{
  Marco~Carbone
  \qquad
  Davide~Grohmann
  \qquad
  Thomas~T.~Hildebrandt
  \qquad
  Hugo~A.~L\'{o}pez
  \institute{IT University of Copenhagen,
    Rued Langgaards Vej 7, 2300 K{\o}benhavn S, Denmark}
  \email{\{carbonem,davg,hilde,lopez\}@itu.dk} 
}

\maketitle

\begin{abstract}
  We explore logical reasoning for the global calculus, a coordination   model based on the notion of choreography, with the aim to provide a  methodology for specification and verification of structured  communications.  Starting with an extension of Hennessy-Milner  logic, we present the global logic (\GL\!\!\!), a modal logic  describing possible interactions among participants in a  choreography.  We illustrate its use by giving examples of  properties on service specifications. Finally, we show that, despite  \GL\!\! is undecidable, there is a significant decidable fragment  which we provide with a sound and complete proof system for checking  validity of formulae.
\end{abstract}


\section{Introduction}

Due to the continuous growth of technologies, software development is
recently shifting its focus on communication, giving rise to various
research efforts for proposing new methodologies dealing with higher
levels of complexity. A new software paradigm, known as {\em
  choreography}, has emerged with the intent to ease programming of
communication-based protocols. Intuitively, a choreography is a
description of the global flow of execution of a system where the
software architect just describes which and in what order interactions
\emph{can} take place. This idea differs from the standard approach
where the communication primitives are given for each single entity
separately. A good illustration can be seen in the way a soccer match
is planned: the coach has an overall view of the team, and organises
(a priori) how players will interact in each play (the r\^ole of a
choreography); once in the field, each player performs his role by
interacting with each of the members of his team by throwing/receiving
passes. The way each player synchronise with other members of the team
represents the r\^ole of an orchestration.

The work in \cite{CHY:dcm2006} formalises the notion of choreography
in terms of a calculus, dubbed the {\em global calculus}, which
pinpoints the basic features of the choreography paradigm. Although
choreography provides a good abstraction of the system being designed
allowing to {\em forget} about common problems that can arise when
programming communication (e.g. races over a channel), it can still have
complex structures hence being often error prone. Additionally,
choreography can be non-flexible in early design stages where the
architect might be interested in designing only parts of a system as
well as specifying only parts of a protocol (e.g. initial and final
interactions). In this view, we believe that a logical approach can
allow for more modularity in designing systems e.g. providing partial
specification of a system using the choreography paradigm.

In order to illustrate the approach proposed in this work, let us
consider an online booking scenario.  On one side, consider an airline company
AC which offers flights directly from its website. On the other side,
there is a customer looking for the best offers.  We can informally
describe the interaction protocol in terms of a sequence of allowed
interactions (as in a choreography) as follows:\\
\begin{tabular}{rl}
  1. & Customer establishes a communication with AC; \\  
  2. & Customer asks AC for a flight proposal given a set of constraints;\\
  3. & AC establishes a communication with partner AC' serving the
  destination asked by the costumer; \\
  4. & AC forwards the request made by the customer; \\
  5. & AC' sends an offer to AC;\\
  6. & AC forwards the offer to the customer
\end{tabular}

\smallbreak




\NI Note that each step above represents a communication. In the same
way that a choreographical specification describes each of the
interactions between participants, a logical characterisation of
choreographies denotes formulae describing the evolution of such
interactions. 
However, a logical characterisation gives extra
flexibility to the specification of interactions:
When writing a logical property describing specific communication patterns we
focus on describing {\bf only} the sequence of \emph{key
  interactions}, leaving room for implementations that include extra
behaviour that does not compromise the fulfilment of  the property.
For instance, in the above example, one can describe a
property leaving out the details on the forward of the request to the
airline partner, in a statement like \emph{``given an interaction
  between the customer and AC featuring a booking request, then there is an
  eventual response directed to the customer with an offer matching
  the original session''} (in this case, the offer is not necessarily from
the airline originally contacted but from one of its partners).


  In this document, we provide a link between choreographies and
  logics. Starting with an extension of Hennessy-Milner logic
  \cite{hennessy1980observing}, we provide the syntax and the
  semantics of a logic for the global calculus as well as several
  examples of choreographical properties. On decidability issues, we
  found out that the whole set of the logic is undecidable on the
  global calculus with recursion. As a result, we focus our studies in
  a decidable fragment, providing a proof system that allows for
  property verification of choreographies and show that it is sound
  and complete, in the sense that all and only valid formulae
  specified in the global logic can be provable in the proof
  system. Moreover, we can conclude that the proof checking algorithm
  using this proof system is terminating.

\paragraph{Overview of the document} First, in Section
\ref{sec:globalCalc} we recall the formal foundations of the global
calculus, and equip it with a labelled transition semantics. A logic
characterisation of the calculus and several examples of the use of
the logic are  presented in Section \ref{sec:globalLogic}. We proceed with
the study of undecidability for the logic in Section
\ref{sec:undecidability}, and a proof system relating the logical
characterisation and the global calculus for a decidable fragment of
the language is presented in Section \ref{sec:proofSys}.  Finally,
concluding remarks are presented in Section \ref{sec:conclusion}.


\section{The Global Calculus}\label{sec:globalCalc}

The Global Calculus (GC) \cite{CHY:dcm2006,carbone7scc} originates
from the Web Service Choreography Description Language (WS-CDL)
\cite{kavantzas2004web}, a description language for web services
developed by W3C. Terms in GC describe choreographies as interactions
between participants by means of message exchanges. The description of
such interactions is centred on the notion of a \emph{session}, in
which two interacting parties first establish a private connection via
some public channel and then interact through it, possibly interleaved
with other sessions. More concretely, an interaction between two
parties starts by the creation of a fresh session identifier, that
later will be used as a private channel where meaningful interactions
take place. Each session is fresh and unique, so each communication
activity will be clearly separated from other interactions. In this
section, we provide an operational semantics for GC in terms of a
label transition systems (LTS) \cite{plotkin81structural} describing
how global descriptions evolve, and relate to the type discipline that describes
the structured sequence of message exchanges between participants from
\cite{carbone7scc}.

\subsection{Syntax}

Let $\chor, \chor', \ldots$ denote {\em terms} of the calculus, often
called {\em interactions} or {\em choreographies}; $A, B, C, \ldots$
range over {\em participants}; $k,k', \ldots$ are {\em linear
  channels}; $a,b,c,\ldots$ {\em shared channels}; $v,w,\ldots$
variables; $X,Y,\ldots$ process variables; $l,l_i,\ldots$ labels for
branching; and finally $e, e', \ldots$ over unspecified arithmetic and
other first-order expressions. We write $e@A$ to mean that the
expression $e$ is evaluated using the variable related to participant
$A$ in the store.

\begin{definition}\label{definition:globalCalc}
  The syntax of the global calculus \cite{CHY:dcm2006} is given by the
  following grammar:
  \begin{alignat}{3}
    \chor ::= \ & \phantom{{}\mid{}}
    &\ & \INACT                          \label{eq:Ginaction} \tag{inaction} \\
    &\mid && \init{A}{B}{a}{k}\pfx \chor \label{eq:Ginit} \tag{init} \\
    &\mid && \interact{A}{B}{k}{e}{y}\pfx\chor  \label{eq:Gcom} \tag{com} \\
    &\mid && \choice{A}{B}{k}{l}{\chor}   \label{eq:Gchoice} \tag{choice} \\
    &\mid && \chor_1 \pp \chor_2          \label{eq:Gpar} \tag{par} \\
    &\mid && \itn{e@A}{\chor_1}{\chor_2}  \label{eq:Gcond} \tag{cond}\\
    &\mid && X                            \label{eq:Grecvar} \tag{recvar} \\
    &\mid && \mu X\pfx C                  \label{eq:Grec} \tag{recursion}
  \end{alignat}
\end{definition}
\NI Intuitively, the term (\ref{eq:Ginaction}) denotes a system where
no interactions take place. (\ref{eq:Ginit}) denotes a session
initiation by $A$ via $B$'s service channel $a$, with a fresh session
channel $k$ and continuation $\chor$. Note that $k$ is bound in
$\chor$.  (\ref{eq:Gcom}) denotes an in-session communication of the
evaluation (at $A$'s) of the expression $e$ over a session channel
$k$. In this case, $y$ does not bind in $\chor$ (our semantics will
treat $y$ as a variable in the store of $B$).  (\ref{eq:Gchoice})
denotes a labelled choice over session channel $k$ and set of labels
$I$. In (\ref{eq:Gpar}), $\chor_1 \pp \chor_2$ denotes the parallel
product between $\chor_1$ and $\chor_2$.  (\ref{eq:Gcond}) denotes the
standard conditional operator where $e@A$ indicates that the
expression $e$ has to be evaluated in the store of participant $A$.
In (\ref{eq:Grec}), $\mu X\pfx\chor$ is the minimal fix point
operation for recursion, where the variable $X$ of (\ref{eq:Grecvar})
is bound in $\chor$.  The free and bound session channels and term
variables are defined in the usual way.  The calculus is equipped with
a standard structural congruence $\equiv$, defined as the minimal
congruence relation on interactions $\chor$, such that $\equiv$ is a
commutative monoid with respect to $\pp$ and $\INACT$, it is closed
under alpha equivalence $\equiv_\alpha$ of terms, and it is closed
under the recursion unfolding, i.e., $ \mu X.\chor \equiv \chor
\MSUBS{\mu X. \chor}{X}$.

\begin{remark}[Differences with the approach in \cite{carbone7scc}]
  \label{remark::one}
  Excluding the lack of local assignment, we argue that this monadic version of
  GC is, to some extent, as expressive as the one Global Calculus originally reported
  in \cite{carbone7scc}.  In particular, note that
  $\interact{A}{B}{k}{\mathsf{op},e}{y}$ in \cite{carbone7scc}
  captures both selection and message passing which are instead
  disentangled in our case (mainly for clarity reasons). The absence
  of $\mathsf{op}$ in the interaction process
  $\interact{A}{B}{k}{e}{y}$ can be easily encoded with the existing
  operators. In fact, $\Sigma_{i\in I}\interact{A}{B}{k}{op_i,e}{y} \pfx
  \chor'_i$ can be decomposed into $\choice{A}{B}{k}{op}{\chor''}$ where $\chor''_i=
  \interact{A}{B}{k}{e}{y} \pfx \chor'_i$ 
  (although we
  lose atomicity).
\end{remark}

\subsection{Semantics}

We give the operational semantics in terms of configurations $(\sigma,
\chor)$, where $\sigma$ represents the state of the system and $\chor$
the choreography actually being executed. The state $\sigma$ contains
a set of variables labelled by participants. As described in the
previous subsection, a variable $x$ located at participant $A$ is
written as $x@A$.  The same variable name labelled with different
participant names denotes different variables (hence $\sigma(x@A$) and
$\sigma(x@B)$ may differ).  Formally, the operational semantics is
defined as a labelled transition system (LTS). A transition
$(\sigma,\chor) \action{\ell} (\sigma',\chor')$ says that a
choreography $\chor$ in a state $\sigma$ executes an action (or label)
$\ell$ and evolves into $\chor'$ with a new state $\sigma'$. Actions
are defined as
$\ell=\{\initF{A}{B}{a(k)},\comF{A}{B}{k},\branchF{A}{B}{k}{l_i}\}$,
denoting initiation, in-session communication and branch selection,
respectively.  We write $(\sigma, \chor) \action{} (\sigma', \chor')$
when $\ell$ irrelevant, and $\action{}^*$ denotes the transitive
closure of $\action{}$. The transition relation $\action{}$ is defined
as the minimum relation on pairs state/interaction satisfying the
rules in Table~\ref{table:global:semantics}.

\begin{table}
  \begin{gather*}
    \Did{G-Init}\
    \inferenceg{} { h \text{ fresh} } { ({\sigma,\init
        {A}{B}{a}{k}\pfx \chor}) \action{\initF{A}{B}{a(h)}}
      ({\sigma,\chor[h/k]}) }
    \\[3mm]
    \Did{G-Com}\ 
    \inferenceg{} {\sigma(e@A)\Downarrow v} {
      ({\sigma,\interact{A}{B}{k}{e}{x}\pfx \chor})
      \action{\comF{A}{B}{k}} ({\sigma[x@B \mapsto v],\chor}) }
    \\[3mm]
    \Did{G-Choice}\ 
    \inferenceg{} {}
    {({\sigma,\choice{A}{B}{k}{l}{\chor}})
      \action{\branchF{A}{B}{k}{l_i}} ({\sigma,\chor_i})}
    \\[3mm]
    \Did{G-Par}\ 
    \inferenceg{} {({\sigma,\chor_1}) \action{\ell}
      ({\sigma',\chor_1'})} {({\sigma,\chor_1\pp \chor_2})
      \action{\ell} ({\sigma',\chor_1'\pp \chor_2})}
    \\[3mm]
    \Did{G-Struct}\ 
    \inferenceg{} {\chor\equiv \chor' \quad
      ({\sigma,\chor'}) \action{\ell} ({\sigma', \chor''}) \quad
      \chor''\equiv \chor'''} {({\sigma,\chor}) \action{\ell}
      ({\sigma',\chor'''})}
    \\[3mm]
    \Did{G-IfT}\ 
    \inferenceg{} {\sigma(e@A)\Downarrow \true \quad
      (\sigma, \chor_1) \action{\ell} (\sigma', \chor'_1)} { (\sigma,
      \itn{e@A}{\chor_1}{\chor_2}) \action{\ell} (\sigma', \chor'_1) }
    \\[3mm]
    \Did{G-IfF}\ 
    \inferenceg{} {\sigma(e@A)\Downarrow \false \quad
      (\sigma, \chor_2) \action{\ell} (\sigma', \chor'_2)} { (\sigma,
      \itn{e@A}{\chor_1}{\chor_2}) \action{\ell}(\sigma',\chor'_2)}
  \end{gather*}
  \caption{Operational Semantics for the Global Calculus}
  \label{table:global:semantics}
\end{table}

Intuitively, transition \Did{G-Init} describes the evolution of a
session initiation: after $A$ initiates a session with $B$ on service
channel $a$, $A$ and $B$ share the fresh channel $h$ locally.
\Did{G-Com} describes the main interaction rule of the calculus: the
expression $e$ is evaluated into $v$ in the $A$-portion of the state
$\sigma$ and then assigned to the variable $x$ located at $B$
resulting in the new state $\sigma [x@B \mapsto v]$. \Did{G-Choice}
chooses the evolution of a choreography resulting from a labelled
choice over a session key $k$. \Did{G-IfT} and \Did{G-IfF} show the
possible paths that a deterministic evolution of a choreography can
produce. \Did{G-Par} and \Did{G-Struct} behave as the standard rules
for parallel product and structural congruence, respectively.

\begin{remark}[Global Parallel]
  Parallel composition in the global calculus differs from the notion
  of parallel found in standard concurrency models based on
  input/output primitives \cite{milner:99:cmspc}. In the latter, a
  term $P_1\pp P_2$ may allow {\em interactions} between $P_1$ and
  $P_2$. However, in the global calculus, the parallel composition of
  two choreographies $\chor_1\pp \chor_2$ concerns two parts of the
  described system where {\em interactions} may occur in $\chor_1$ and
  $\chor_2$ but never across the parallel operator $\pp$. This is
  because an interaction $A\rightarrow B\ldots$ abstracts from the
  actual end-point behaviour, i.e., how $A$ sends and $B$ receives. In
  this model, dependencies between two choreographies can be expressed
  by using variables in the state $\sigma$.
\end{remark}

In its original presentation \cite{carbone7scc}, GC comes equipped
with a reduction semantics unlike the one presented in Table
\ref{table:global:semantics}. Our LTS semantics has the advantage of
allowing to observe changes on the behaviour of the system, which will
prove useful when relating to the logical characterisation in Section
\ref{sec:globalLogic}. We conjecture that our proposed LTS semantics
and the reduction semantics of the global calculus originally
presented in \cite{carbone7scc} coincide (taking into account the
considerations in Remark~\ref{remark::one}).

\begin{example}[Online Booking]\label{ex:onlinebooking}
  We consider the example presented in the introduction, i.e., a
  simplified version of the on-line booking scenario presented in
  \cite{LOP-places09}. Here, the customer (Cust) establishes a session
  with the airline company (AC) using service (on-line booking,
  shorted as ob) and creating the session key $k_1$. Once the session
  is established, the customer will request the company about a
  flight offer with his booking data, along the session key $k_1$. The
  airline company will process the customer request and, after
  requesting another airline company (AC') for the service, will send
  a reply back with an offer. The customer will eventually accept the
  offer, sending back an acknowledgment to the airline company using
  $k_1$. The following specification in the GC represents the
  protocol:
  \begin{align} \label{example:syntax} 
    \chor_{\mathsf{OB}} = {} &
    \init{\text{Cust}}{\text{AC}}{\text{ob}}{k_1} \pfx
    \interact{\text{Cust}}{\text{AC}}{k_1}{\text{booking}}{x} \pfx 
    \init{\text{AC}}{\text{AC'}}{\text{ob}}{k_2} \pfx
    \tag{OB}\\
    & 
    \interact{\text{AC}}{\text{AC'}}{k_2}{\text{x}}{x'} \pfx 
    \interact{\text{AC'}}{\text{AC}}{k_2}{\text{offer}}{y} \pfx
    \interact{\text{AC}}{\text{Cust}}{k_1}{\text{y}}{y''} \pfx
    \interact{\text{Cust}}{\text{AC}}{k_1}{\text{accept}}{z} \pfx
    \INACT
    \notag{}
  \end{align}
\end{example}

\subsection{Session Types for the Global Calculus}


We use a generalisation of session types \cite{honda1998lpa} for
global interactions, first presented in \cite{carbone7scc}.  Session
types in GC are used to structure sequence of message exchanges in a
session.  Their syntax is as follows:
\begin{equation}
  \alpha =  
  \uparrow(\theta).\alpha
  ~|~
  \downarrow(\theta).\alpha
  ~|~
  \&\{ l_i: \alpha_i\}_{i\in{I}}
  ~|~ \oplus\{ l_i: \alpha_i\}_{i\in{I}}  ~|~  \endT ~|~ \mu \mathbf{
    t }\pfx \alpha ~|~  \mathbf{ t }
\end{equation}
where $\theta, \theta', \dots$ range over value types $\mathtt{bool,
  string, int, \dots}$. $\alpha, \alpha', \dots $ are session
types. The first four types are associated with the various
communication operations.  $\downarrow(\theta).\alpha$ and
$\uparrow(\theta).\alpha$ are the input and output types
respectively. Similarly, $\&\{ l_i: \alpha_i\}_{i\in{I}}$ is the
branching type while $\oplus\{ l_i: \alpha_i\}_{i\in{I}}$ is the
selection type. The type $\endT$ indicates session termination and is
often omitted.  $\mu \mathbf{t} \pfx \alpha$ indicates a recursive
type with $\mathbf{t}$ as a type variable. $\mu \mathbf{t} \pfx
\alpha$ binds the free occurrences of $\mathbf{t}$ in $\alpha$. We
take an \emph{equi-recursive} view on types, not distinguishing
between $\mu \mathbf{t} \pfx \alpha$ and its unfolding $\alpha [\mu
\mathbf{t} \pfx \alpha / \mathbf{t}]$.

A typing judgment has the form $\typerule{\typeEnv}{}{\chor :
  \Delta}$, where $\typeEnv, \Delta$ are \emph{service type} and
\emph{session type} environments, respectively. Typically, $\typeEnv$
contains a set of type assignments of the form $a@A: \alpha$, which
says that a service $a$ located at participant $A$ may be invoked and
run a session according to type $\alpha$. $\Delta$ contains type
assignments of the form $k[A,B]: \alpha$ which says that a session
channel $k$ identifies a session between participants $A$ and $B$ and
has session type $\alpha$ when seen from the viewpoint of $A$. The
typing rules are omitted, and we refer to \cite{carbone:tbc} for the
full account of the type discipline noting that the observations made
in Remark~\ref{remark::one} will require extra typing rules.

Returning to the specification (\ref{example:syntax}) in
Example~\ref{ex:onlinebooking}, the service type of the airline
company AC at channel $ob$ can be described as:
\begin{equation*}
  \text{ob}@\text{AC}:(k_1,k_2) ~ k_1 \downarrow
  \text{booking}(\mathtt{string}) \pfx~ k_2 \uparrow
  \text{x}(\mathtt{string}) \pfx k_2 \downarrow
  \text{offer}(\mathtt{int}) \pfx k_1 \uparrow
  \text{y}(\mathtt{int})\pfx k_1 \downarrow \text{accept}
  (\mathtt{int}) \pfx \endT \, .
\end{equation*}


\begin{assumption}
  In the sequel, we only consider choreographies that satisfy the
  typing discipline.
\end{assumption}


\section{\GL: A Logic for the Global Calculus}\label{sec:globalLogic}

In this section, we introduce a logic for choreographies, inspired by
the modal logic for session types presented in
\cite{Berger2008Completeness-an}.  The logical language comprises
assertions for equality and value/name passing.

\subsection{Syntax}
\begin{table}
  \hspace{-1cm}
  \begin{minipage}{.5\textwidth}
    \begin{alignat}{3}
      \phi, \chi \ ::= \ & \phantom{{} \mid {}}
      & \ & \exists \var\pfx \phi \label{eq:exists} \tag{f-exists}\\
      & \mid && \phi \land \chi \label{eq:and} \tag{f-and} \\
      & \mid && \neg \phi \label{eq:neg} \tag{f-neg} \\
      & \mid && \langle\ell\rangle \phi \label{eq:action} \tag{f-action} \\
      & \mid && \endF \label{eq:termination} \tag{f-termination} \\
      & \mid && e_1@A = e_2@B \label{eq:equality} \tag{f-equality} \\
      & \mid && \phi \mid \chi \label{eq:parallel} \tag{f-parallel} \\
      & \mid && \may \phi \label{eq:may} \tag{f-may} 
    \end{alignat}
  \end{minipage}
  \hfill
  \begin{minipage}{.5\textwidth}
    \vspace{-2.9cm}
    \begin{alignat}{3}
      \ell \ ::= \ & \phantom{{}\mid{}}
      &\ & \initF A B {a(k)} \label{eq:init} \tag{l-init} \\
      & \mid && \comF ABk && \label{eq:com} \tag{l-com}\\
      & \mid && \branchF ABkl & \label{eq:branch} \tag{l-branch}
    \end{alignat}
  \end{minipage}
  \caption{\GL: Syntax of formulae}
  \label{table:ChoreographyLogic}
\end{table}
The grammar of assertions is given in
Table~\ref{table:ChoreographyLogic}.  Choreography assertions (ranged
over by $\phi, \phi', \chi, \dots$) give a logical interpretation of
the global calculus introduced in the previous section.  The logic
includes the standard First Order Logic (FOL) operators $\land$, $\neg$, and $\exists$. In
$\exists \var \pfx \phi$, the variable $\var$ is meant to range over
service and session channels, participants, labels for branching and
basic placeholders for expressions.  Accordingly, it works as a binder
in $\phi$.  In addition to the standard operators, the operator
(\ref{eq:action}) represents the execution of a labelled action $\ell$
followed by the assertion $\phi$. Those labels in $\ell$ match the ones
in the LTS of GC, i.e., they are (\ref{eq:init}), (\ref{eq:com}), and
(\ref{eq:branch}).  The formula (\ref{eq:termination}) represents the
process termination.  We also include an unspecified, but decidable,
(\ref{eq:equality}) operator on expressions as in
\cite{Berger2008Completeness-an}.  (\ref{eq:may}) denotes the standard
eventually operators from Linear Temporal Logic (LTL)
\cite{emerson1991temporal}. The spatial operator (\ref{eq:parallel})
denotes composition of formulae: because of the unique nature of
parallel composition in choreographies, we have used the symbol $\mid$
(as in separation logic \cite{reynolds2002sll} and spatial logic
\cite{caires2001spatial}) in order to stress the fact that there is no
interference between two choreographies running in
parallel.

\begin{notation}[Existential quantification over action labels]
  In order to simplify the readability, we introduce the concept of
  existential quantification over action labels as a short-cut to mean
  the following:
  \begin{alignat*}{2}
    \exists \ell \pfx \langle \ell \rangle \phi \ \DEFEQ {} \
    &&& \exists A,B,a,k \pfx
    \langle \initF{A}{B}{a(k)} \rangle \pfx \phi \lor {} \\
    &&& \exists A,B,k \pfx
    \langle \comF{A}{B}{k} \rangle \pfx \phi \lor {} \\
    &&& \exists A,B,k,l \pfx
    \langle \branchF{A}{B}{k}{l} \rangle \pfx \phi \, .
  \end{alignat*}
\end{notation}

\begin{remark}[Derived Operators]
  We can get the full account of the logic by deriving the standard
  set of strong modalities from the above presented operators. In
  particular, we can encode the constant true ($\true$) and false
  ($\false$), the next ($\nextOp \phi$) and the always operators
  ($[] \phi$) from LTL.
  \begin{alignat*}{3}
    \true & \DEFEQ (0@A = 0@A) &\qquad\qquad
    \false  & \DEFEQ (0@A = 1@A) &\qquad
    (e_1 \neq e_2) & \DEFEQ \neg (e_1 = e_2) \\
    \forall x \pfx \phi & \DEFEQ \neg \exists x  \pfx \lnot \phi &
    \phi \lor \chi & \DEFEQ \lnot ( \lnot \phi  \land \lnot \chi) &
    \phi => \chi & \DEFEQ \lnot \phi \lor \chi \\
    [] \phi & \DEFEQ \lnot <<>> \neg \phi &
    [\ell] \phi & \DEFEQ \neg \langle \ell \rangle \neg \phi &
    \nextOp \phi & \DEFEQ \exists \ell \pfx \langle \ell \rangle \phi \, . 
  \end{alignat*}
\end{remark}

In the rest of this section, we illustrate the expressiveness of our
logic through a sequence of simple, yet illuminating examples, giving
an intuition of how the modalities introduced plus the existential
operator $\exists$ allow to express properties of choreographies.

\begin{example}[Availability, Service Usage and Coupling]
  The logic above allows to express that, given a service invoker
  (known as $A$ in this setting) requesting the service $a$, there
  exists another participant (called $B$ in the example) providing $a$
  with $A$ invoking it.  This can be formulated in \GL as follows:
  \begin{equation*}
    \exists B \pfx \langle \initF{A}{B}{a(k)}\rangle \true \, .
  \end{equation*}
  Assume now, that we want to ensure that services available are
  actually used. We can use the dual property for availability, i.e.,
  for a service provider $B$ offering $a$, there exists someone
  invoking $a$:
  \begin{equation*}
    \exists A \pfx \langle \initF{A}{B}{a(k)}\rangle \true \, .
  \end{equation*}
  Verifying that there is a service pairing two different participants
  in a choreography can be done by existentially quantifying over the
  shared channels used in an initiation action. A formula in \GL
  representing this can be the following one:
  \begin{equation*}
    \exists a \pfx \langle \initF{A}{B}{a(k)} \rangle \true \, .
  \end{equation*}
\end{example}
\begin{example}[Causality Analysis]
  The modal operators of the logic can be used to perform studies of
  the causal properties that our specified choreography can fulfil.
  For instance, we can specify that given an expression $e$ evaluated
  to true at participant $A$, there is an eventual firing of a
  choreography that satisfies property $\phi_1$, whilst $\phi_2$ will
  never be satisfied.  Such a property can be specified as follows:
  \begin{equation*}
    (e@A = \true) \land <<>> (\phi_1) \land [] \lnot \phi_2 \, .
  \end{equation*}
\end{example}
\begin{example}[Response Abstraction]
  \begin{figure}
    \begin{center}
      \includegraphics[width=\textwidth]{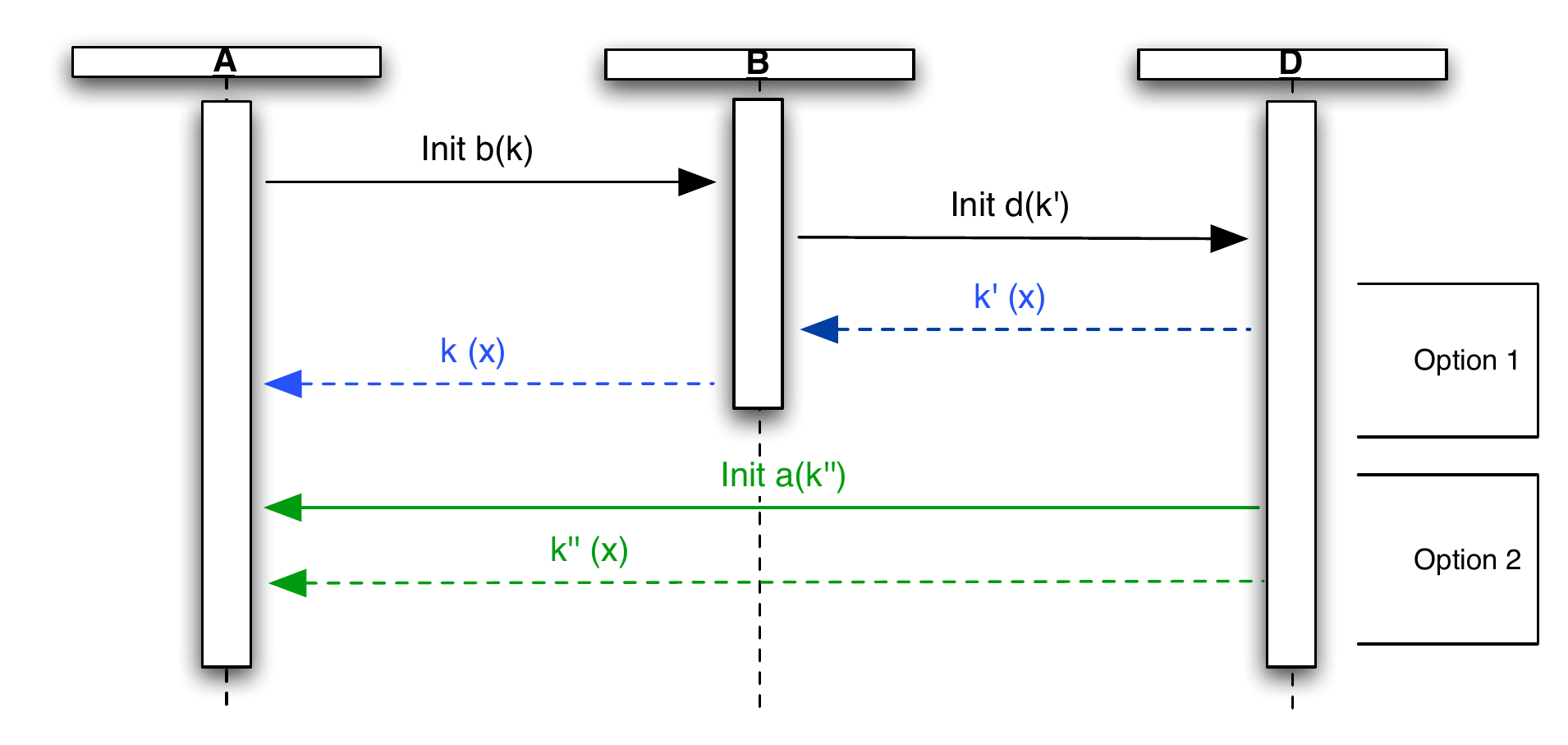}
    \end{center}
    \caption{Diagram of a partial specification.}
    \label{fig:diagram}
  \end{figure}
  An interesting aspect of our logic is that it allows for the
  declaration of partial specification properties regarding the
  interaction of the participants involved in a choreography. Take for
  instance the interaction diagram in Figure~\ref{fig:diagram}.  The
  participant $A$ invokes service $b$ at $B$'s and then $B$ invokes
  $D$'s service $d$.  At this point, $D$ can send the content of
  variable $x$ to $A$ in two different ways: either by using those
  originally established sessions or by invoking a new service at
  $A$'s. However, at the end of both computation paths, variable $z$
  (located at $A$'s) will contain the value of $x$. In the global
  calculus, this two optional behaviour can be modelled as follows:
  \begin{align}
    C_1 & = \init{A}{B}{b}{k} \pfx \init{B}{D}{d}{k'} \pfx
    \interact{D}{B}{k'}{x}{y_B}\pfx\interact{B}{A}{k}{y_B}{z} \pfx \INACT 
    \label{eq:option1} \tag{Option 1} \\
    C_2 & = \init{A}{B}{b}{k} \pfx \init{B}{D}{d}{k'} \pfx
     \init{D}{A}{a}{k''}\pfx \interact{D}{A}{k''}{x}{z} \, \pfx \INACT \, .
     \label{eq:option2} \tag{Option 2}
  \end{align}
  We argue that, under the point of view of $A$, both options are
  sufficiently good if, after an initial interaction with $B$ is
  established, there is an eventual response that binds variable
  $z$. Such a property can be expressed by the \GL formula:
  \begin{equation*}
    \exists X,{k''}\pfx \langle\initF{A}{B}{a(k)} \rangle
    \may \Big(\langle \comF{X}{A}{k''}\rangle (z@A=x@D) \Big) \pfx \endT \, .
  \end{equation*}
  Notice that both the choreographies (\ref{eq:option1}) and
  (\ref{eq:option2}) \emph{satisfy} the partial specification
  above. This will be clear in Section~\ref{logic-assertions} where we
  introduce the semantics of logic.

  Also note that a third option for the protocol at hand is to use
  \emph{delegation} (the ability of communicating session keys to
  third participants not involved during session initiation). However,
  the current version of the global calculus does not feature such an
  operation and we leave it as future work.
\end{example}

\begin{example}[Connectedness]
  The work in \cite{carbone7scc} proposes a set of criteria for
  guaranteeing a safe end-point projection between global and local
  specifications (note that the choreography in the previous example
  does not respect such properties). Essentially, a valid global
  specification has to fulfil three different criteria, namely
  Connectedness, Well-threadedness and Coherence.  It is interesting
  to see that some of these criteria relate to global and local
  causality relations between the interactions in a choreography, and
  can be easily formalised as properties in the choreography logic
  presented here. Below, we consider the notion of connectedness and
  leave the other cases as future work. Connectedness dictates a
  global causality principle among interactions: any two consecutive
  interactions $\ldots A\rightarrow B \pfx C\rightarrow D\ldots$ in a
  choreography are such that $B=C$.  
  In the following, let $\mathsf{Interact}(A,B)\phi$ be true whenever
  $\langle\ell\rangle\phi$ holds for some $\ell$ with an interaction
  from $A$ to $B$.  Connectedness can be specified as:
  \begin{equation*}
    \forall A,B\pfx
    []
    \Big(
    \mathsf{Interact}(A,B)\true \Rightarrow {}  
    \exists C\pfx 
    \big(\mathsf{Interact}(A,B)\mathsf{Interact}(B,C)\true
    \lor 
    \mathsf{Interact}(A,B)\neg\exists\ell\langle\ell\rangle\true\big)
    \Big) \, . 
  \end{equation*}
\end{example}

\subsection{Semantics}\label{logic-assertions}

\begin{table}
  \centering
  \begin{displaymath}
    \begin{array}{lcl}
      \chor\ |=_\sigma \endF & \defSym & \chor \equiv \INACT 
      \\
      \chor\ |=_\sigma (e_1@A = e_2@B) & \defSym & \sigma(e_1@A)\Downarrow v\text{ and }\sigma(e_2@B)\Downarrow v 
      \\
      \chor\ |=_{\sigma}\actionF{ \ell}\phi&\defSym &  (\sigma,\chor)\action{\ell}(\sigma',\chor')\text{ and }\chor'|=_{\sigma'} \phi
      \\
      \chor\ |=_\sigma \phi\land\chi&\defSym &\chor|=_\sigma \phi \text{ and }\chor|=_\sigma \chi 
      \\
      \chor\ |=_\sigma \neg \phi & \defSym & \chor \not |=_\sigma \phi     
      \\
      \chor\ |=_\sigma \exists\var\pfx\phi & \defSym & \chor |=_\sigma \phi[w/\var]
      \text{ (for some appropriate $w$)}
      \\
      \chor\ |=_\sigma \may \phi & \defSym & (\sigma,\chor) \action{}^*(\sigma',\chor') \text{ and } \chor' |=_{\sigma'} \phi
      \\
      \chor\ |=_\sigma \phi \pp \chi & \defSym & \chor\ \equiv\ \chor_1\pp \chor_2 \text{ such that }
      \chor_1 |=_\sigma\phi \text{ and } \chor_2 |=_\sigma \chi
    \end{array}
  \end{displaymath}
  \caption{Assertions of the Choreography Logic}
  \label{table:global:assertions}
\end{table}

We now give a formal meaning to the assertions introduced above with
respect to the semantics of the global calculus introduced in the
previous section. In particular, we introduce the notion of
satisfaction.  We write $\chor |=_\sigma \phi$ whenever a state
$\sigma$ and a choreography $\chor$ satisfy a \GL formula $\phi$.  The
relation $|=_\sigma$ is defined by the rules given in
Table~\ref{table:global:assertions}.  In the $\exists\var\pfx\phi$
case, $w$ should be an appropriate value according to the type of
$\var$, e.g., a participant if $\var$ is a participant placeholder.

\begin{definition}[Satisfiability, Validity and Logical Equivalence]\
  \begin{itemize}
  \item A formula $\phi$ is \emph{satisfiable} if there exists some
    configuration under which it is true, that is, $\chor |=_\sigma
    \phi$ for some $(\chor,\sigma)$.
  \item A formula $\phi$ is \emph{valid} if it is true in every
    configuration, that is, $\chor |=_\sigma \phi$ for every
    $(\sigma,\chor)$.
  \item A formula $\chi$ is a \emph{logical consequence} of a formula
    $\phi$ (or $\phi$ \emph{logically implies} $\chi$), denote with
    an abuse of notation as $\phi |= \chi$, if every configuration
    $(\sigma,\chor)$ that makes $\phi$ true also makes $\chi$ true.
  \item We say that a formula $\phi$ is \emph{logical equivalent} to
    a formula $\chi$, written $\phi \logicEquiv \chi$, if $\phi |= \chi$ iff
    $\chi |= \phi$.
  \end{itemize}
\end{definition}



\section{Undecidability of Global Logic} \label{sec:undecidability}
In this section we focus on the undecidability of the global logic for
the global calculus with recursion given in
Section~\ref{sec:globalCalc}.  In order to prove that the global logic
is undecidable, we use a reduction from the Post Correspondence
Problem (PCP)~\cite{Post:pcp} similarly to the one proposed in
\cite{ct:csl01}. The idea is to encode in the global calculus a
``program'' which simulates the construction of PCP.
We first give a formal definition of the PCP.  In the sequel, $\cdot$
denotes word concatenation.
\begin{definition}[PCP]
  Let $s,t,\ldots$ range over $\Sigma^*$ where $\Sigma = \{0, 1\}$ and
  let $\epsilon$ be the empty word.  An instance of PCP is a set of
  pairs of words $\{(s_1,t_1), \ldots, (s_n,t_n)\}$ over
  $\Sigma^*\times\Sigma^*$.  The Post Correspondence Problem is to
  find a sequence $i_0,i_1,\dots,i_k$ ($1 \leq i_j \leq n$ for all
  $0\leq j \leq k$) such that $s_{i_0}\cdot \ldots \cdot s_{i_k} =
  t_{i_0}\cdot \ldots \cdot t_{i_k}$.
\end{definition}

\NI Intuitively, PCP consists of finding some string in $\Sigma^*$
which can be obtained by the concatenation $s_{i_0}\cdot \ldots \cdot
s_{i_k}$ as well as by $t_{i_0}\cdot \ldots \cdot t_{i_k}$. Such a
problem has been proved to be undecidable~\cite{Post:pcp}.
Our goal is to find a GC term that takes a random pair of words from
an instance of PCP and append them to an ``incremental pair'' of words
which encodes the current state of the sequences $s_{i_0}\cdot \ldots
\cdot s_{i_k}$ and $t_{i_0}\cdot \ldots \cdot t_{i_k}$.  Technically,
we need a choreography that assigns randomly a natural number in $\{1,
\dots, n\}$ to a variable $r$ in some participant $B$, and another
choreography that picks a pair of words from the PCP instance,
accordingly to value in the variable $r@B$, and then appends them to
the ``incremental pair'' of words in $A$.
Formally, 
\begin{definition}[Encoding of PCP]\label{def:PCPencoding}
  Let $A_1,\dots,A_n,A,B$ be participants and $a,b$ shared names for
  sessions, then define the two choreographies as shown below:
  \begin{align*}
    &
    \begin{array}{rcl}
      \textsf{Random}(A_1,\dots,A_n,B,a) &
      \DEFEQ &
      \phantom{{}\pp\ {}}\mu X \pfx
      \init{A_1}{B}{a}{k} \pfx
      \interact{A_1}{B}{k}{1}{r} \pfx X\\[1mm]
      &&
      \pp\ \mu X \pfx
      \init{A_2}{B}{a}{k} \pfx
      \interact{A_2}{B}{k}{2}{r} \pfx X\\[1mm]
      &&
      \pp\ \ldots\\[1mm]
      &&
      \pp\ \mu X \pfx
      \init{A_n}{B}{a}{k} \pfx
      \interact{A_n}{B}{k}{n}{r} \pfx X
      \\[2ex]
      \textsf{Append}(A,B,b) &
      \DEFEQ &
      \mu X\pfx
      \init{A}{B}{b}{k} \pfx
      \interact{A}{B}{k}{str1}{tmp1} \pfx
      \interact{A}{B}{k}{str2}{tmp2} \pfx {} \\
      &
      &
      \mathbf{if}\ r@B = 1\ \mathbf{then} \\
      &
      &
      \quad
      \interact{B}{A}{k}{tmp1 \cdot s_1}{str1} \pfx
      \interact{B}{A}{k}{tmp2 \cdot t_1}{str2} \pfx X \\
      &
      &
      \mathbf{else}\ \mathbf{if}\ r@B = 2\ \mathbf{then} \\
      &
      &
      \quad
      \interact{B}{A}{k}{tmp1 \cdot s_2}{str1} \pfx
      \interact{B}{A}{k}{tmp2 \cdot t_2}{str2} \pfx X \\
      &
      &
      \mathbf{else} \ \mathbf{if}\ r@B = 3\ \mathbf{then} \\
      &
      &
      \qquad \vdots \\
      &
      &
      \mathbf{else} \ \mathbf{if}\ r@B = n\ \mathbf{then} \\
      &
      &
      \quad
      \interact{B}{A}{k}{tmp1 \cdot s_n}{str1} \pfx
      \interact{B}{A}{k}{tmp2 \cdot t_n}{str2} \pfx X \\
      &
      &
      \mathbf{else }\ X
    \end{array}
  \end{align*}
  We define the initial configuration $(\sigma,\chor)$ to be formed by
  the choreography and the state below:
  {\small
  \begin{align*}
    \chor & \DEFEQ
    \textsf{Random}(A_1,\dots,A_n,B,a) \pp  \textsf{Append}(A,B,b) \\
    \sigma & \DEFEQ
    [str1@A \mapsto \epsilon,\ str2@A \mapsto \epsilon,\
    tmp1@B \mapsto \epsilon,\ tmp2@B \mapsto \epsilon,\
    r@B \mapsto 1] \, .
  \end{align*}
  }
  For encoding the PCP existence question ($s_{i_0}\cdot \ldots \cdot
  s_{i_k} = t_{i_0}\cdot \ldots \cdot t_{i_k}$) we can encode it as a
  \GL formula:
  {\small
  \begin{equation*}
    \phi \DEFEQ
    <<>> \Big(
    (str1@A = str2@A) \land
    (str1@A \neq \epsilon) \land
    (str2@A \neq \epsilon)
    \Big) \, .
  \end{equation*}
  }
\end{definition}
\NI Above, each participant $A_i$ (with $i\in \{1,\dots,n\}$)
recursively opens a session with participant $B$ and writes in the
variable $r@B$ the value $i$. Moreover, the participant $B$ stores the
knowledge of all the word pairs $(s_i,t_i)$, while the participant $A$
takes randomly a word pair from $B$ and then append it to his
incremental pair of words: $(str1,str2)$.  Next, the formula $\phi$
states that there exists a computational path from the initial
configuration to a configuration which stores in $str1$ and $str2$ two
equal non-empty strings.

\begin{theorem}
  The global logic is undecidable on the global calculus with
  recursion.
\end{theorem}
\begin{proof} 
  (Sketch) The statement $\chor |=_\sigma \phi$ holds iff the encoded
  PCP has a solution.  Indeed, if the initial configuration
  $(\sigma,\chor)$ satisfies the formula $\phi$ then it means there
  exists a configuration $(\sigma',\chor')$ where $(str1@A = str2@A)
  \land (str1@A \neq \epsilon) \land (str2@A \neq \epsilon)$
  holds. Hence, there is a sequence of $i_0,\dots,i_k$ such that $str1
  = s_{i_0}\cdot \ldots \cdot s_{i_k} = t_{i_0}\cdot \ldots \cdot
  t_{i_k} = str2$, that is, the instance of PCP has a solution.
\end{proof}

\begin{remark}\label{remark:vars}
  The undecidability result presented in this section shows that the
  global calculus is considerably expressive, despite the choreography
  approach offers a simplification in the specification of concurrent
  communicating systems as argued in \cite{carbone7scc}. The encoding
  in Definition~\ref{def:PCPencoding} shows that allowing state
  variables (hence local variables that can be accessed by various
  threads) increases the expressive power of the language. Indeed, we
  could just look at GC as a simple concurrent language with a
  ``shared'' store where assignment to variables is just in-session
  communication. In this view, we conjecture that removing variables
  and focusing only on communication would make the logic decidable.
\end{remark}


\section{Proof System for Recursion-free Choreographies}\label{sec:proofSys}

In this section, we present a model checking algorithm (in the form of
a proof system) to decide whether a global logic formula is satisfied by
a recursion-free configuration of the global calculus. Indeed,
similarly to \cite{ct:csl01}, it turns out that the logic is decidable
on the recursion-free choreographies.\footnote{Removing recursion
  yields a decidability result orthogonal to the conjecture formulated
  in Remark~\ref{remark:vars}} We also prove the soundness and
completeness of the proposed proof system w.r.t.~the assertion
semantics.

In order to reason about judgments $\chor |=_\sigma \phi$, we propose
a proof (or inference) system for assertions of the form $\chor
|-_\sigma \phi$.  Intuitively, we want $\chor |-_\sigma \phi$ to be as
approximate as possible to $\chor |=_\sigma \phi$ (ideally, they
should be equivalent).  We write $\chor |-_\sigma \phi $ for the
provability judgement where $(\sigma,\chor)$ is a configuration and
$\phi$ is a formula.  

\begin{notation}
  We define the set of continuations configuration after an action
  $\ell$ and the reachable configurations, both starting from a
  configuration $(\sigma, \chor)$, as follows:
  \begin{align*}
    \textsf{Next}(\sigma,\chor,\ell) & \DEFEQ
    \{(\sigma',\chor') \mid (\sigma,\chor) \action{\ell} (\sigma',\chor')\}
    \\
    \textsf{Reachable}(\sigma,\chor) & \DEFEQ
    \{(\sigma',\chor') \mid (\sigma,\chor) \action{}^* (\sigma',\chor')\}
    \, .
  \end{align*}
Normalisation is required by the proof system to infer equality of
choreographies up to structural equi\-va\-len\-ce (Especially for the
$[ \cdot ] \pp [ \cdot ]$ operator).
  We define $\textsf{Norm}(\chor)$ to be a normalisation function from
  recursion-free choreographies into multi-sets of choreographies:
  \begin{align*}
    & 
    \textsf{Norm}(\interact{A}{B}{k}{e}{y}\pfx \chor) \DEFEQ
    [\interact{A}{B}{k}{e}{y}\pfx \chor]
    \qquad
    \textsf{Norm}(\choice{A}{B}{k}{l}{\chor}) \DEFEQ
    [\choice{A}{B}{k}{l}{\chor}]
    \\
     & \textsf{Norm}(\init{A}{B}{a}{k}\pfx \chor) \DEFEQ
    [\init{A}{B}{a}{k}\pfx \chor]
    \qquad\!\!\!
    \textsf{Norm}(\ifthenelse{e@A}{\chor_1}{\chor_2}) \DEFEQ
    [\ifthenelse{e@A}{\chor_1}{\chor_2}]    
    \\
    &
    \textsf{Norm}(\INACT) \DEFEQ [\ ]
    \qquad
    \textsf{Norm}(\chor_1 \pp \chor_2) \DEFEQ [P_1,\dots,P_n,Q_1,\dots,Q_m]
    \quad \text{if }
    \begin{array}{ll}
      \textsf{Norm}(\chor_1) = [P_1,\dots,P_n] & \text{and} \\
      \textsf{Norm}(\chor_2) = [Q_1,\dots,Q_m] & .
    \end{array}
  \end{align*}
\end{notation}

\begin{lemma}[Normalisation preserves structural equivalence]
  \label{lem:normalization}
  Let $\chor$ be a recursion-free choreography and
  $\textsf{Norm}(\chor) = [P_1,\dots,P_n]$, then $\chor \equiv
  \prod_{i=1}^n P_i$.
\end{lemma}
\begin{proof}
  By induction on the structure of the choreography $\chor$.
  \begin{description}
  \item[Case $\chor = \INACT$:] We have $\textsf{Norm}(\INACT) = [\
    ]$, and $\prod_{i=1}^0 P_i = \INACT \equiv \INACT$.
  \item[Case $\chor = \chor_1 \pp \chor_2$:] We have that
    $\textsf{Norm}(\chor_1) = [P_1,\dots,P_n]$,
    $\textsf{Norm}(\chor_2) = [Q_1,\dots,Q_m]$, and $\prod_{i=1}^n P_i
    \equiv \chor_1$, $\prod_{j=1}^m Q_j \equiv \chor_2$ by induction
    hypothesis. Then, we can derive that $\prod_{i=1}^n P_i \pp
    \prod_{j=1}^m Q_j \equiv \chor_1 \pp \chor_2$.
  \item[All the other cases:] Trivially we have that
    $\textsf{Norm}(\chor) = [P_1]$, where $P_1 = \chor$, then
    $\prod_{i=1}^1 P_i \equiv \chor$. \qedhere
  \end{description}

\end{proof}

\begin{definition}[Entailment]
  We say that a choreography $\chor$ \emph{entails} a formula $\phi$
  under a state $\sigma$, written $\chor|-_\sigma \phi$, iff the
  assertion $\chor|-_\sigma \phi$ has a proof in the proof system
  given in Table~\ref{table:Global:proofSys}.
\end{definition}

\begin{table}
  \begin{gather*}
    \myruleg{P_{end}}{\textsf{Norm}(\chor) = [\ ]}
    {\typerule{\chor}{\sigma}{\endF}}
    \qquad\qquad
    \myruleg{P_{and}}{\typerule{\chor}{\sigma}{\phi} \quad
      \typerule{\chor}{\sigma}{\chi}} {\typerule{\chor}{\sigma}{\phi
        \land \chi}}
    \qquad\qquad
    \myruleg{P_{neg}}{ \chor \not |-_{\sigma} \phi}
    {\typerule{\chor}{\sigma}{\neg \phi}}
    \\[1ex]
    \myruleg{P_{par}}{
        \textsf{Norm}(\chor) = [P_1,\dots,P_n] \quad
        \exists I,J.\
        I\cup J = \{1,\dots,n\} \wedge
        I\cap J = \emptyset \wedge
        \typerule{\prod_{i\in I} P_i}{\sigma}{\phi_1} \wedge
        \typerule{\prod_{j\in J} P_j}{\sigma}{\phi_2}
     }
    {\typerule{\chor}{\sigma}{ \phi_1 \pp \phi_2}}
    \\[1ex]
    \myruleg{P_{action}}{\exists (\sigma',\chor') \in
      \textsf{Next}(\sigma,\chor,\ell).\  \typerule{\chor'}{\sigma'}{\phi}}
    {\typerule{\chor}{\sigma}{\langle\ell \rangle \phi}}
    \qquad\qquad
    \myruleg{P_{may}}{\exists (\sigma',\chor') \in
      \textsf{Reachable}(\sigma,\chor).\  \typerule{\chor'}{\sigma'}{\phi}}
    {\typerule{\chor}{\sigma}{\may \phi}}
    \\[1ex]
    \myruleg{P_{\exists}}
    {\exists w\in fn(\chor)\cup fn(\phi).\ 
      \typerule{\chor}{\sigma}{\phi[w/\var]}}
    {\typerule{\chor}{\sigma}{\exists \var \pfx \phi}}
    \qquad\qquad
    \myruleg{P_{exp}}{
      \sigma(e_1 @ A) \Downarrow v \quad \sigma(e_2 @ B) \Downarrow v}
    {\typerule{\chor}{\sigma}{ (e_1 @ A = e_2 @ B) }}
  \end{gather*}
  \caption{Proof system for the Global Calculus.}
  \label{table:Global:proofSys}
\end{table}

Let us now describe some of the inference rules of the proof system.
The rule $\mathsf{P_{end}}$ relates the inaction terms with the
termination formula. The rules $\mathsf{P_{and}}$ and
$\mathsf{P_{neg}}$ denote rules for conjunction and negation in
classical logic, respectively. The rule for parallel composition is
represented in $\mathsf{P_{par}}$; it does not indicate the behaviour
of a given choreography, but hints information about the structure of
the process: $\mathsf{P_{par}}$ juxtaposes the behaviour of two
processes and combines their respective formulae by the use of a
separation operator. The next rule, $\mathsf{P_{action}}$ requires
that the process $P$ in the configuration $\sigma$ can perform an
action labelled $\ell$, so we must search for a continuations of
$(\sigma,\chor)$ after an action $\ell$ and find a configuration which
satisfies the rest of the formula, i.e., $\phi$.  Analogously,
$\mathsf{P_{may}}$ looks for a continuation in the reachable
configuration of $(\sigma,\chor)$ in oder to satisfy $\phi$.  The rule
$\mathsf{P_\exists}$ says that in order to satisfy an $\exists t\pfx
\phi$, it is sufficient to find a value $w$ for $t$ in the free names
used by the choreography $\chor$ or in the free names used by the
formula $\phi$. Finally, the rule $\mathsf{P_{exp}}$ denotes
evaluation of expressions.

We now proceed to prove the soundness of the proof system with respect
to the semantics of assertions presented before.

\begin{lemma}[Structural congruence preserves satisfability]
  \label{lemma:StructuralSatisfability} If $\chor\equiv\chor'$ and
  $\chor|=_\sigma \phi$, then $\chor' |=_\sigma \phi$.
\end{lemma}
\begin{proof}
  (Sketch) It follows from structural induction over $\phi$.
\end{proof}

\begin{theorem}[Soundness]\label{thm:soundness}
  For any configuration $(\sigma,\chor)$, where $\chor$ is
  recursion-free, and every formula $\phi$, if $\chor|-_\sigma \phi$
  then $\chor|=_\sigma \phi$.
\end{theorem}
\begin{proof}
  It follows by induction on the derivation of $|-_\sigma$.
 \begin{description}
 \item[Case $\mathsf{P_{end}}$:] Straight consequence of
   Lemmas~\ref{lem:normalization} and
   \ref{lemma:StructuralSatisfability}, indeed $\chor \equiv \INACT$
   and $\chor |=_\sigma \endT$.
 \item[Case $\mathsf{P_{and}}$:] By induction hypothesis and
   conjunction.
 \item[Case $\mathsf{P_{neg}}$:] We have that $\chor|-_\sigma \lnot
   \phi$, so by $\mathsf{P_{neg}}$ we get $\chor \not |-_\sigma
   \phi$. By induction hypothesis we have that $\chor \not |=_\sigma
   \phi$, which is the necessary condition to deduce $\chor |=_\sigma
   \lnot \phi$.
 \item[Case $\mathsf{P_{par}}$:] We have that $\chor |-_\sigma
   \phi_1\pp \phi_2$, then $\textsf{Norm}(\chor) = [P_1,\dots,P_n]$,
   and there exist $I,J$ such that $I\cup J = \{1,\dots,n\}$, $I\cap J
   = \emptyset$, $\typerule{\prod_{i\in I} P_i}{\sigma}{\phi_1}$, and
   $\typerule{\prod_{j\in J} P_j}{\sigma}{\phi_2}$. By induction
   hypothesis we know that $\prod_{i\in I} P_i |=_\sigma \phi_1$ and
   $\prod_{j\in J} P_j |=_\sigma \phi_2$, then by
   Lemma~\ref{lem:normalization} we have $\chor \equiv \prod_{i\in I}
   P_i \pp \prod_{j\in J} P_j$, hence it is immediate to prove that
   $\chor |=_\sigma \phi_1 \pp \phi_2$.
 \item[Case $\mathsf{P_{action}}$:] We have that $\chor |-_\sigma
   \langle \ell \rangle \phi$ and by $\mathsf{P_{action}}$ then
   $\chor' |-_{\sigma'} \phi$ and $(\sigma', \chor') \in
   \textsf{Next}(\sigma, \chor, \ell)$. From the induction hypothesis
   we have that $\chor' |=_{\sigma'} \phi$, then we have to show that
   $\chor |=_\sigma \rangle \ell \langle \phi$. From the assertion
   semantics we know that $C |=_\sigma \langle \ell \rangle \phi$ iff
   $(\sigma, \chor') \action{\ell} (\sigma', \chor')$ and $\chor'
   |=_{\sigma'} \phi$, which holds immediately by the selection of
   $(\sigma', \chor') \in \textsf{Next}(\sigma, \chor,\ell)$ and the
   induction hypothesis.
 \item[Case $\mathsf{P_{may}}$:] We have that $\chor |-_\sigma \may
   \phi$ and by $\mathsf{P_{may}}$ then $\chor' |-_{\sigma'} \phi$ and
   $(\sigma', \chor') \in \textsf{Reachable}(\sigma, \chor)$. From the
   induction hypothesis we have that $\chor' |=_{\sigma'} \phi$, then
   we have to show that $\chor |=_\sigma \may \phi$. From the
   assertion semantics we know that $C |=_\sigma \may \phi \iff
   (\sigma, \chor') \action{} ^* (\sigma', \chor')$ and $\chor'
   |=_{\sigma'} \phi$, which holds immediately by the selection of
   $(\sigma', \chor') \in \textsf{Reachable}(\sigma, \chor)$ and the
   induction hypothesis.
 \item[Case $\mathsf{P_{\exists}}$:] We have that $\chor |-_\sigma
   \exists t. \phi$ and by $\mathsf{P_{\exists}}$ we have that
   $\exists w \in fn(\chor) \cup fn(\phi)$ and $\chor |-_\sigma \phi
   [w/t]$. By induction hypothesis we know that $C |=_\sigma \phi
   [w/t]$ with appropriate $w \in fn(\chor) \cup fn(\phi)$, then
   $\chor |=_\sigma \exists t. \phi$ follows from the definition of
   the assertion semantics.
 \item[Case $\mathsf{P_{exp}}$:] It holds trivially by checking if
   $\sigma(e_1@A) \Downarrow v$ and $\sigma(e_2@B) \Downarrow
   v$. \qedhere
 \end{description}
\end{proof}

\begin{lemma}\label{lem:exists}
  For every configuration $(\sigma,\chor)$, where $\chor$ is recursion
  free, and every formula $\exists t\pfx \phi$, if $\{n_1,\dots,n_k\}
  = fn(\chor) \cup fn(\phi)$, then $\chor |=_\sigma \exists t\pfx
  \phi$ iff $\exists m\in \{n_1,\dots,n_k\}$ such that $\chor
  |=_\sigma \phi[m/t]$.
\end{lemma}
\begin{proof}
  (Sketch) By induction on the structure of $\phi$.  It is similar to
  the proof of \cite[Lemma~5.3(3)]{cg:popl00}.
\end{proof}

\begin{theorem}[Completeness]\label{thm:completeness}
  For any configuration $(\sigma,\chor)$, where $\chor$ is
  recursion-free, and every formula $\phi$, if $\chor|=_\sigma \phi$
  then $\chor|-_\sigma \phi$.
\end{theorem}
\begin{proof}
  By rule induction on the derivation of $|=_\sigma$.
  \begin{description}
  \item[Case $\chor |=_\sigma \endT$:] We have that $\chor \equiv
    \INACT$ and hence $\textsf{Norm}(\chor) = [\ ]$ by
    Lemma~\ref{lem:normalization}. Now, the thesis follows immediately
    from the application of $ \mathsf{P_{end}}$.
  \item[Case $\chor |=_\sigma (e_1 @ A = e_2 @ B)$:] It follows
    immediately by the application of $\mathsf{P_{exp}}$.
  \item[Case $\chor |=_\sigma \langle \ell \rangle \phi'$:] Take
    $(\sigma, \chor) \action{\ell} (\sigma',\chor')$ and $\chor'
    |=_{\sigma'} \phi' $, we have by induction hypothesis that
    $\typerule{\chor'}{\sigma'}{\phi'}$. Now, we have to show that
    $\chor |-_\sigma \langle \ell \rangle \phi'$.  By the fact that
    $(\sigma, \chor) \action{\ell} (\sigma',\chor')$, we have that
    $(\sigma', \chor')\in \textsf{Next}(\sigma,\chor,\ell)$, hence, we
    can apply rule $\mathsf{P}_{action}$ and we are done.
  \item[Case $\chor |=_\sigma \phi \land \chi$:] We have that $\chor
    |=_\sigma \phi$ and $\chor |=_\sigma \chi$. From the induction
    hypothesis we have that $\chor |-_\sigma \phi$ and $ \chor
    |-_\sigma \chi$. The application of $\mathsf{P_{and}}$ lead to
    $\chor |-_\sigma \phi \land \chi$ as desired.
  \item[Case $\chor |=_\sigma \lnot \phi$:] From the definition of the
    assertion semantics we have that $\chor |=_\sigma \lnot \phi $ iff
    $\chor \not |=_\sigma \phi$. We have to show that $\chor |-_\sigma
    \lnot \phi$. We proceed by contradiction. Take a $(\phi, \chor)$
    such that $\chor |-_\sigma \phi$, then from
    Theorem~\ref{thm:soundness} we have that $\chor |=_\sigma \phi$,
    which is a contradiction to $\chor |=_\sigma \lnot \phi$.
  \item[Case $\chor |=_\sigma \exists \var \pfx \phi$:] We have that
    $\chor |=_\sigma \exists t. \phi$ and by the definition in the
    assertion semantics we have that $\chor |=_\sigma \phi [w/t]$ for
    an appropriate $w$. By induction hypothesis we know that $\chor
    |-_\sigma \phi[w/t]$. Lemma~\ref{lem:exists} guarantees that there
    exists $w \in fn(\chor)\cup fn(\phi)$ in order to derive $\chor
    |-_\sigma \exists t. \phi$ from $\mathsf{P_{\exists}}$.
  \item[Case $\chor |=_\sigma <<>> \phi$:] Take $(\sigma, \chor)
    \action{}^* (\sigma',\chor')$ and $\chor' |=_{\sigma'} \phi'$, we
    have by induction hypothesis that
    $\typerule{\chor'}{\sigma'}{\phi'}$. Now, we have to show that
    $\chor |-_\sigma <<>> \phi'$.  By the fact that $(\sigma, \chor)
    \action{}^* (\sigma',\chor')$, we have that $(\sigma', \chor')\in
    \textsf{Reachable}(\sigma,\chor)$, hence, we can apply rule
    $\mathsf{P}_{may}$ and we are done.
  \item[Case $\chor |=_\sigma \phi \pp \chi$:] We have that $\chor
    \equiv \chor_1 \pp \chor_2$ and $\chor_1 |=_\sigma \phi \land
    \chor_2 |=_\sigma \chi$. From the induction hypothesis $\chor_1
    |-_\sigma \phi$ and $\chor_2 |-_\sigma \chi$. Now by
    Lemma~\ref{lem:normalization} we have that $\chor_1 \equiv
    \prod_{i\in I} P_i$ and $\chor_2 \equiv \prod_{j\in J} P_j$ for
    some $I,J$. So, we can derive $\chor \equiv \prod_{i\in I} P_i \pp
    \prod_{j\in J} P_j$, and hence $\mathsf{P_{par}}$ leads to
    $\chor_1 \pp \chor_2 |-_\sigma \phi \pp \chi$. \qedhere
  \end{description} 
\end{proof}

\begin{theorem}[Termination]\label{thm:termination}
  For any configuration $(\sigma,\chor)$, where $\chor$ is
  recursion-free, and every formula $\phi$, proof-checking algorithm
  terminates.
\begin{proof}
  First, notice that all the functions \textsf{Norm}, \textsf{Next},
  and \textsf{Reachable} are total and computable. The proof is by
  induction over the structure of $\phi$.
  \begin{description}
  \item[Case $\phi = \endT$:] $\typerule{\chor}{\sigma}{\endT}$ iff
    $\textsf{Norm}(\chor) = [\ ]$.
  \item[Case $\phi = \phi_1 \land \phi_2$:] By conjunction and
    induction hypothesis on $\typerule{\chor}{\sigma}{\phi_1}$ and
    $\typerule{\chor}{\sigma}{\phi_2}$.
  \item[Case $\phi = \neg \phi'$:] $\typerule{\chor}{\sigma}{\phi}$
    iff $\typerule{\chor}{\sigma}{\phi'}$ does not hold. But by
    induction hypothesis we can construct a terminating proof or
    confutation for $\typerule{\chor}{\sigma}{\phi'}$. Hence the proof
    for $\typerule{\chor}{\sigma}{\phi}$ terminates as well.
  \item[Case $\phi = \phi_1 \pp \phi_2$:] Suppose
    $\textsf{Norm}(\chor) = [P_1,\dots,P_n]$. Notice that there exists
    a finite number of possible partitioning of $\{1,\dots,n\}$ in
    $I,J$. Hence, for every $I,J$ we can compute
    $\typerule{\prod_{i\in I} P_i}{\sigma}{\phi_1}$ and
    $\typerule{\prod_{j\in J} P_j}{\sigma}{\phi_2}$, which both
    terminate by induction hypothesis. By applying
    Lemma~\ref{lem:normalization} we prove the thesis.
  \item[Case $\phi = \langle \ell \rangle \phi'$:] First, notice that
    the set $\textsf{Next}(\sigma,\chor,\ell)$ is finite, because the
    choreographies are finite, i.e., there are a finite number of
    actionable transition in a given configuration. For each
    configuration $(\sigma',\chor') \in
    \textsf{Next}(\sigma,\chor,\ell)$,
    $\typerule{\chor'}{\sigma'}{\phi'}$ terminates by induction
    hypothesis.
  \item[Case $\phi = <<>> \phi'$:] As before, notice that the set
    $\textsf{Reachable}(\sigma,\chor)$ is finite, because the
    choreographies are finite, i.e., the choreographies are recursion
    free. For each configuration $(\sigma',\chor') \in
    \textsf{Reachable}(\sigma,\chor)$,
    $\typerule{\chor'}{\sigma'}{\phi'}$ terminates by induction
    hypothesis.
  \item[Case $\phi = \exists t\pfx \phi'$:] To prove existence is
    sufficient to check every derivation by substituting $t$ with a
    name $w\in fn(\chor)\cup fn(\phi)$. Notice that $fn(\chor) \cup
    fn(\phi)$ is finite, because both $\chor$ and $\phi$ are so. So,
    for every $w$, we can construct a terminating derivation for
    $\typerule{\chor}{\sigma}{\phi'[w/t]}$ by induction hypothesis.
  \item[Case $\phi = (e_1@A = e_@@B):$]
    $\typerule{\chor}{\sigma}{(e_1@A = e_@@B)}$ iff $e_1@A \Downarrow
    v$ and $e_@@B \Downarrow v$.  \qedhere
  \end{description}
\end{proof}
\end{theorem}


\section{Conclusion and Related Work}
\label{sec:conclusion}

The ideas hereby presented constitutes just the first step towards a
verification framework for choreography. As a future work, our main
concerns relate to integrate our framework into other end-point models
and logical frameworks for the specification of sessions. In
particular, our next step will focus on relating the logic to the
end-point projection \cite{carbone7scc}, the process of automatically
generating end-point code from choreography. Other improvements to the
system proposed include the use of fixed points, essential for
describing state-changing loops, and auxiliary axioms
describing 
structural properties of a choreography. 

This work can be fruitfully nourished by related work in types and
logics for session-based communication. In \cite{LOP-places09} the
authors proposed a mapping between the calculus of structured
communications and concurrent constraint programming, allowing them to
establish a logical view of session-based communication and formulae
in First-Order Temporal Logic. In \cite{Berger2008Completeness-an},
Berger et al. presented proof systems characterising May/Must testing
pre-orders and bisimilarities over typed \mipi-calculus processes. The
connection between types and logics in such system comes in handy to
restrict the shape of the processes one might be interested, allowing
us to consider such work as a suitable proof system for the calculus
of end points. Finally, \cite{semini} studies a logic for
choreographies in a model without services and sessions while
\cite{Bocchi2010A-theory-of-des} proposes notion of global assertion
for enriching multiparty session types with simple formula describing
changing in the state of a session.  
\vspace{-0.4cm}
\paragraph{Acknowledgements}
This research has been partially supported by the Trustworthy
Pervasive Healthcare Services (TrustCare) and the Computer Supported
Mobile Adaptive Business Processes (Cosmobiz) projects. Danish Research
Agency, Grants $\#$ 2106-07-0019 (www.TrustCare.eu) and $\#$
274-06-0415 (www.cosmobiz.org).

\label{sect:bib}
\bibliographystyle{eptcs}
\bibliography{biblio}


\end{document}